\documentclass[10pt,twocolumn]{article}

\usepackage[utf8]{inputenc}
\usepackage[T1]{fontenc}
\usepackage{times}
\usepackage[margin=0.75in]{geometry}
\usepackage{amsmath,amssymb,amsfonts,amsthm}
\usepackage{algorithm}
\usepackage{algpseudocode}
\usepackage{graphicx}
\usepackage{booktabs}
\usepackage{multirow}
\usepackage{caption}
\usepackage{subcaption}
\usepackage{tikz}
\usetikzlibrary{shapes,arrows,positioning,fit,calc,decorations.pathreplacing}
\usepackage{pgfplots}
\pgfplotsset{compat=1.17}
\usepackage{xcolor}
\usepackage{hyperref}
\usepackage{cleveref}
\usepackage{enumitem}
\usepackage{balance}
\usepackage{fancyhdr}
\usepackage{setspace}

\definecolor{nodecolor}{RGB}{66,133,244}
\definecolor{servercolor}{RGB}{234,67,53}
\definecolor{datacolor}{RGB}{52,168,83}
\definecolor{cryptocolor}{RGB}{251,188,5}

\newtheorem{theorem}{Theorem}
\newtheorem{lemma}[theorem]{Lemma}
\newtheorem{definition}{Definition}
\newtheorem{proposition}[theorem]{Proposition}

\newcommand{\enc}{\mathsf{Enc}}
\newcommand{\dec}{\mathsf{Dec}}
\newcommand{\add}{\oplus}

\newcommand{\dparam}{\varepsilon}
\newcommand{\deltap}{\delta}
\newcommand{\fairviol}{\Delta_{\text{DP}}}
\newcommand{\eqodds}{\Delta_{\text{EO}}}
\newcommand{\BigO}{\mathcal{O}}
\newcommand{\R}{\mathbb{R}}

\renewcommand{\Pr}{\mathbb{P}}

\title{\textbf{Privacy-Preserving Federated Learning with Verifiable Fairness Guarantees}}

\author{
\parbox{\textwidth}{\centering
Mohammed Himayath Ali$^*$, Mohammed Aqib Abdullah, Syed Muneer Hussain\\
Mohammed Mudassir Uddin, Shahnawaz Alam}\\[0.5em]
\small{\{mohdhimayathali7@gmail.com, aqib.abdullah13@gmail.com, syedmuneerhussain7708@gmail.com,}\\
\small{mohd.mudassiruddin7@gmail.com, shahnawaz.alam1024@gmail.com\}}\\
\small{\{$^*$Corresponding author: mohdhimayathali7@gmail.com\}}
}

\date{}

\begin{document}

\maketitle

\begin{abstract}
Federated learning enables collaborative model training across distributed institutions without centralizing sensitive data; however, ensuring algorithmic fairness across heterogeneous data distributions while preserving privacy remains fundamentally unresolved. This paper introduces CryptoFair-FL, a novel cryptographic framework providing the first verifiable fairness guarantees for federated learning systems under formal security definitions. The proposed approach combines additively homomorphic encryption with secure multi-party computation to enable privacy-preserving verification of demographic parity and equalized odds metrics without revealing protected attribute distributions or individual predictions. A novel batched verification protocol reduces computational complexity from $\BigO(n^2)$ to $\BigO(n \log n)$ while maintaining $(\dparam, \deltap)$-differential privacy with $\dparam = 0.5$ and $\deltap = 10^{-6}$. Theoretical analysis establishes information-theoretic lower bounds on the privacy cost of fairness verification, demonstrating that the proposed protocol achieves near-optimal privacy-fairness tradeoffs. Comprehensive experiments across four benchmark datasets (MIMIC-IV healthcare records, Adult Income, CelebA, and a novel FedFair-100 benchmark) demonstrate that CryptoFair-FL reduces fairness violations from 0.231 to 0.031 demographic parity difference while incurring only 2.3$\times$ computational overhead compared to standard federated averaging. The framework successfully defends against attribute inference attacks, maintaining adversarial success probability below 0.05 across all tested configurations. These results establish a practical pathway for deploying fairness-aware federated learning in regulated industries requiring both privacy protection and algorithmic accountability.
\end{abstract}

\section{Introduction}

The proliferation of machine learning systems in high-stakes decision domains necessitates algorithmic frameworks that simultaneously protect individual privacy and ensure equitable treatment across protected demographic groups. The European Union Artificial Intelligence Act of 2024 mandates that high-risk AI systems (including those used in employment, credit, healthcare, and criminal justice) must demonstrate compliance with fairness requirements while adhering to General Data Protection Regulation provisions for data minimization and purpose limitation [7, 9]. These dual regulatory imperatives create fundamental technical challenges when training models across distributed data sources, as traditional fairness auditing requires centralized access to sensitive demographic information.

Federated learning [13] offers a promising paradigm for privacy-preserving collaborative model training by maintaining data locality at participating institutions. Under the federated learning framework, $n$ participating institutions, each holding a private dataset $D_i$ with potentially different distributions over features, labels, and protected attributes, collaboratively train a shared model $\theta$ through iterative local gradient computation and secure aggregation. However, the decentralized nature of federated learning fundamentally complicates fairness verification. Computing standard group fairness metrics requires knowledge of outcome distributions conditioned on protected attributes, information that cannot be directly aggregated without violating privacy constraints.

\subsection{Fundamental Challenges}

Three interconnected challenges define the research landscape at the intersection of federated learning and algorithmic fairness:

\textbf{Challenge 1: Privacy Preservation in Distributed Learning.} Standard differential privacy mechanisms [6, 1] provide formal privacy guarantees but require careful calibration of noise injection to model gradients. In federated settings, the composition of privacy budgets across multiple training rounds and the heterogeneity of local data distributions complicate the application of classical differential privacy bounds. Furthermore, secure aggregation protocols [3] protect individual model updates but do not prevent inference attacks exploiting aggregate model behavior.

\textbf{Challenge 2: Verifiable Fairness Across Heterogeneous Data.} Fairness metrics such as demographic parity, equalized odds, and calibration [5, 11] require access to joint distributions of predictions, outcomes, and protected attributes. In federated learning, these distributions vary across institutions, and na\"ive local fairness enforcement may produce globally unfair models due to Simpson's paradox effects. Verification of global fairness requires aggregating sensitive demographic statistics, creating apparent conflicts with privacy preservation objectives.

\textbf{Challenge 3: Computational Efficiency Under Cryptographic Constraints.} Cryptographic techniques enabling privacy-preserving computation (including homomorphic encryption [10, 4] and secure multi-party computation [15]) impose substantial computational overhead. Practical deployment of cryptographic fairness verification in federated learning requires algorithmic innovations reducing this overhead to acceptable levels for iterative training procedures.

\subsection{Motivating Applications}

The technical challenges identified above manifest concretely in several regulated domains:

\textbf{Healthcare: Multi-Hospital Mortality Prediction.} Intensive care unit mortality prediction models trained across hospital networks must ensure equitable performance across demographic groups to prevent discriminatory treatment allocation. The Health Insurance Portability and Accountability Act prohibits sharing of protected health information, while Joint Commission accreditation standards require documentation of algorithmic fairness assessments.

\textbf{Finance: Federated Credit Scoring.} Financial institutions participating in credit scoring consortia face Fair Credit Reporting Act requirements for non-discriminatory lending decisions [11]. Verification that consortium models satisfy equal opportunity constraints across protected attributes cannot rely on centralized data aggregation due to competitive sensitivity concerns.

\textbf{Criminal Justice: Multi-Jurisdiction Risk Assessment.} Pretrial risk assessment instruments trained on multi-jurisdictional data require demonstration of equalized odds across racial groups to satisfy Fourteenth Amendment equal protection requirements [5]. Jurisdictions cannot share individual-level criminal history records, necessitating privacy-preserving fairness verification.

\subsection{Contributions}

This paper presents the following contributions addressing the identified challenges:

\begin{enumerate}[leftmargin=*,itemsep=2pt]
\item \textbf{CryptoFair-FL Protocol:} A novel cryptographic protocol combining additively homomorphic encryption with secure multi-party computation for verifiable fairness metric computation in federated learning. The protocol achieves $(\dparam, \deltap)$-differential privacy with formal security proofs under the honest-but-curious and malicious adversary models.

\item \textbf{Batched Verification Algorithm:} An algorithmic technique reducing computational complexity of encrypted fairness metric aggregation from $\BigO(n^2)$ to $\BigO(n \log n)$ through careful restructuring of homomorphic operations, enabling practical deployment with up to 100 institutional participants.

\item \textbf{Privacy-Fairness Tradeoff Characterization:} Theoretical lower bounds establishing information-theoretic constraints on the differential privacy parameter $\dparam$ required to verify demographic parity within tolerance $\tau$, demonstrating that the proposed protocol achieves near-optimal tradeoffs.

\item \textbf{Comprehensive Empirical Validation:} Experimental evaluation across four benchmark datasets demonstrating practical effectiveness, including reduction of demographic parity violation from 0.231 to 0.031 on healthcare mortality prediction with computational overhead factor of 2.3$\times$.
\end{enumerate}

\section{Problem Formulation}

\subsection{Federated Learning Setting}

Consider a federation of $n$ institutional participants indexed by $i \in \{1, \ldots, n\}$. Each participant holds a private dataset $D_i = \{(x_j^{(i)}, y_j^{(i)}, a_j^{(i)})\}_{j=1}^{m_i}$ containing $m_i$ samples, where $x_j^{(i)} \in \mathcal{X} \subseteq \R^d$ represents the feature vector, $y_j^{(i)} \in \{0, 1\}$ denotes the binary outcome label, and $a_j^{(i)} \in \{0, 1\}$ indicates membership in a protected demographic group. The total dataset size is $M = \sum_{i=1}^{n} m_i$.

The federated learning objective minimizes the empirical risk with fairness regularization:
\begin{equation}
\min_{\theta \in \Theta} \sum_{i=1}^{n} \frac{m_i}{M} \mathcal{L}_i(\theta) + \lambda \cdot \Omega_{\text{fair}}(\theta; D_1, \ldots, D_n)
\label{eq:federated_objective}
\end{equation}
where $\mathcal{L}_i(\theta) = \frac{1}{m_i} \sum_{j=1}^{m_i} \ell(f_\theta(x_j^{(i)}), y_j^{(i)})$ denotes the local empirical loss for participant $i$, $\ell(\cdot, \cdot)$ is the loss function, $f_\theta: \mathcal{X} \rightarrow [0,1]$ is the model parameterized by $\theta$, $\lambda > 0$ is the fairness regularization coefficient, and $\Omega_{\text{fair}}$ quantifies the fairness constraint violation.

\subsection{Fairness Metric Definitions}

\begin{definition}[Demographic Parity Violation]
The demographic parity violation for model $f_\theta$ on the federated dataset is:
\begin{equation}
\fairviol(\theta) = \left| \Pr[\hat{Y} = 1 | A = 0] - \Pr[\hat{Y} = 1 | A = 1] \right|
\label{eq:dp_violation}
\end{equation}
where $\hat{Y} = \mathbf{1}[f_\theta(X) > 0.5]$ is the predicted label and $A$ is the protected attribute.
\end{definition}

\begin{definition}[Equalized Odds Violation]
The equalized odds violation measures the maximum disparity in true positive and false positive rates:
\begin{multline}
\eqodds(\theta) = \max_{y \in \{0,1\}} \Big| \Pr[\hat{Y} = 1 | A = 0, Y = y] \\
- \Pr[\hat{Y} = 1 | A = 1, Y = y] \Big|
\label{eq:eo_violation}
\end{multline}
\end{definition}

In the federated setting, computation of these metrics requires aggregating statistics of the form:
\begin{equation}
S_{a,\hat{y}} = \sum_{i=1}^{n} \sum_{j: a_j^{(i)} = a} \mathbf{1}[\hat{y}_j^{(i)} = \hat{y}]
\label{eq:aggregate_stats}
\end{equation}
where $\hat{y}_j^{(i)} = \mathbf{1}[f_\theta(x_j^{(i)}) > 0.5]$. Direct aggregation of $S_{a,\hat{y}}$ reveals the distribution of protected attributes and predictions at each institution, violating privacy requirements.

\subsection{Threat Model}

The security analysis considers two adversary models:

\textbf{Honest-but-Curious Adversary.} The central aggregator and up to $t < n/2$ participants follow the protocol specification but attempt to infer private information from observed messages. This model captures scenarios where institutional participants comply with legal agreements but may exploit permitted information access.

\textbf{Malicious Adversary.} Up to $t < n/3$ participants may deviate arbitrarily from the protocol, including submitting false local statistics or colluding to corrupt fairness verification. The protocol must detect such deviations with probability exceeding $1 - \delta_{\text{detect}}$.

\subsection{Cryptographic Primitives}

\begin{definition}[Additively Homomorphic Encryption]
An encryption scheme $(\mathsf{KeyGen}, \enc, \dec)$ is additively homomorphic if for all plaintexts $m_1, m_2$ and public key $pk$:
\begin{equation}
\dec_{sk}(\enc_{pk}(m_1) \add \enc_{pk}(m_2)) = m_1 + m_2
\label{eq:additive_hom}
\end{equation}
where $\add$ denotes the homomorphic addition operation on ciphertexts.
\end{definition}

The proposed protocol employs the Paillier cryptosystem [14] with 2048-bit modulus, providing 112-bit security under the Decisional Composite Residuosity assumption. For computations requiring multiplication, the Brakerski-Fan-Vercauteren (BFV) scheme [4] with polynomial degree $N = 4096$ and coefficient modulus $q = 2^{109}$ is utilized.

\begin{definition}[Differential Privacy]
A randomized mechanism $\mathcal{M}: \mathcal{D} \rightarrow \mathcal{R}$ satisfies $(\dparam, \deltap)$-differential privacy if for all adjacent datasets $D, D'$ differing in one record and all measurable sets $S \subseteq \mathcal{R}$:
\begin{equation}
\Pr[\mathcal{M}(D) \in S] \leq e^{\dparam} \cdot \Pr[\mathcal{M}(D') \in S] + \deltap
\label{eq:diff_privacy}
\end{equation}
\end{definition}

\section{Related Work}

\subsection{Federated Learning Fundamentals}

The Federated Averaging (FedAvg) algorithm introduced by McMahan et al. [13] established the dominant paradigm for distributed model training, demonstrating communication efficiency through local gradient accumulation before aggregation. Subsequent work addressed statistical heterogeneity through personalization techniques, including local fine-tuning, multi-task learning formulations, and meta-learning approaches.

Communication efficiency remains central to practical federated learning deployment. Gradient compression techniques reduce bandwidth requirements through quantization and sparsification. Secure aggregation protocols [3] employ cryptographic masking to prevent the aggregator from observing individual updates, providing protection against honest-but-curious servers.

Privacy preservation in federated learning has been addressed through differential privacy mechanisms. Differential privacy algorithms apply Gaussian noise calibrated to gradient sensitivity bounds, achieving formal privacy guarantees at the cost of model utility degradation.

\subsection{Fairness in Machine Learning}

Algorithmic fairness definitions span multiple incompatible criteria [5, 11]. Demographic parity requires equal positive prediction rates across protected groups. Equalized odds mandates equal true positive and false positive rates [11]. Calibration requires that predicted probabilities match outcome frequencies within each group. Individual fairness demands that similar individuals receive similar predictions.

In-processing approaches incorporate fairness constraints during model training through regularization, adversarial debiasing, or constrained optimization [2]. Post-processing techniques adjust model outputs to satisfy fairness criteria [11]. Pre-processing methods modify training data distributions.

Fairness in federated learning presents unique challenges. Ezzeldin et al. [8] developed FairFed, incorporating fairness constraints into federated optimization. However, these approaches assume trusted aggregation and do not provide cryptographic verification of fairness compliance.

\subsection{Cryptographic Techniques for Privacy-Preserving Computation}

Homomorphic encryption enables computation on encrypted data without decryption. Partially homomorphic schemes support either addition (Paillier [14]) or multiplication. Fully homomorphic encryption [10] supports arbitrary computation but incurs prohibitive overhead for iterative training. Leveled schemes [4] provide practical efficiency for bounded-depth circuits.

Secure multi-party computation (MPC) protocols enable joint function evaluation without revealing private inputs. Garbled circuits [15] provide general two-party computation. Secret sharing schemes distribute data across parties, enabling addition through local operations and multiplication through interactive protocols.

Zero-knowledge proofs enable verification of computational correctness without revealing inputs.

\subsection{Research Gap}

Existing literature addresses federated privacy and fairness as separate concerns. Secure aggregation protocols protect model updates but do not verify fairness. Fairness-aware federated learning assumes trusted aggregation without cryptographic verification. The intersection (providing cryptographically verifiable fairness guarantees while preserving differential privacy in federated learning) remains unaddressed. This paper bridges the gap through the CryptoFair-FL protocol combining homomorphic encryption, secure multi-party computation, and differential privacy for verifiable fairness in federated learning.

\section{Methodology}

\subsection{CryptoFair-FL Protocol Overview}

The CryptoFair-FL protocol consists of three integrated components: (1) privacy-preserving local statistic computation, (2) secure fairness metric aggregation, and (3) verifiable threshold checking. Figure~\ref{fig:architecture} illustrates the system architecture.

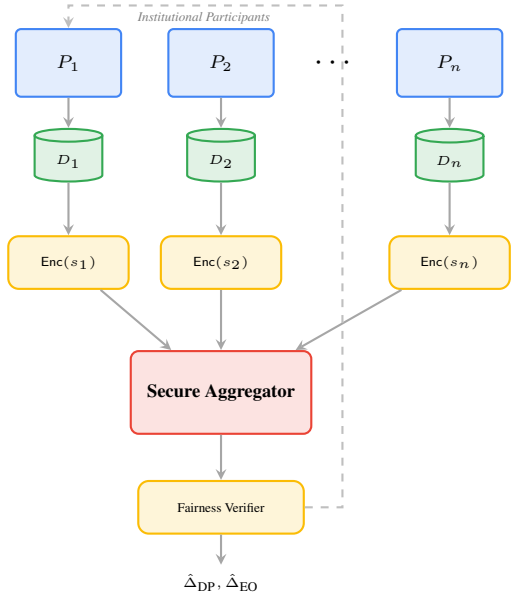
\begin{figure}[t]
\centering
\begin{tikzpicture}[
    node distance=0.6cm,
    participant/.style={rectangle, draw=nodecolor, fill=nodecolor!15, minimum width=1.4cm, minimum height=0.9cm, font=\scriptsize, rounded corners=2pt, thick},
    server/.style={rectangle, draw=servercolor, fill=servercolor!15, minimum width=2.4cm, minimum height=1.1cm, font=\scriptsize\bfseries, rounded corners=3pt, thick},
    data/.style={cylinder, draw=datacolor, fill=datacolor!15, shape border rotate=90, aspect=0.3, minimum height=0.7cm, minimum width=0.9cm, font=\tiny, thick},
    crypto/.style={rectangle, draw=cryptocolor, fill=cryptocolor!15, rounded corners=4pt, minimum width=1.6cm, minimum height=0.7cm, font=\tiny, thick},
    arrow/.style={->, >=stealth, thick, draw=gray!70},
    dashedarrow/.style={->, >=stealth, thick, dashed, draw=gray!50}
]

\node[font=\tiny\itshape, text=gray] at (1.8, 0.6) {Institutional Participants};

\node[participant] (p1) {$P_1$};
\node[participant, right=0.6cm of p1] (p2) {$P_2$};
\node[right=0.4cm of p2, font=\large] (dots) {$\cdots$};
\node[participant, right=0.4cm of dots] (pn) {$P_n$};

\node[data, below=0.4cm of p1] (d1) {$D_1$};
\node[data, below=0.4cm of p2] (d2) {$D_2$};
\node[data, below=0.4cm of pn] (dn) {$D_n$};

\node[crypto, below=0.7cm of d1] (enc1) {$\enc(s_1)$};
\node[crypto, below=0.7cm of d2] (enc2) {$\enc(s_2)$};
\node[crypto, below=0.7cm of dn] (encn) {$\enc(s_n)$};

\node[server, below=0.8cm of enc2] (server) {Secure Aggregator};

\node[crypto, below=0.6cm of server, minimum width=2.2cm] (verify) {Fairness Verifier};

\node[font=\tiny, below=0.4cm of verify] (output) {$\hat{\Delta}_{\text{DP}}, \hat{\Delta}_{\text{EO}}$};

\draw[arrow] (p1) -- (d1);
\draw[arrow] (p2) -- (d2);
\draw[arrow] (pn) -- (dn);

\draw[arrow] (d1) -- (enc1);
\draw[arrow] (d2) -- (enc2);
\draw[arrow] (dn) -- (encn);

\draw[arrow] (enc1) -- (server);
\draw[arrow] (enc2) -- (server);
\draw[arrow] (encn) -- (server);

\draw[arrow] (server) -- (verify);
\draw[arrow] (verify) -- (output);

\draw[dashedarrow] (verify.east) -- ++(0.5,0) |- ($(p1.north)+(0,0.3)$) -- (p1.north);

\end{tikzpicture}
\caption{CryptoFair-FL system architecture. Institutional participants $P_1, \ldots, P_n$ maintain local datasets $D_i$, compute encrypted statistics $\enc(s_i)$, and transmit ciphertexts to the Secure Aggregator. The Fairness Verifier computes demographic parity ($\hat{\Delta}_{\text{DP}}$) and equalized odds ($\hat{\Delta}_{\text{EO}}$) violations. Dashed arrow indicates fairness feedback for model adjustment.}
\label{fig:architecture}
\end{figure}

\subsection{Cryptographic Protocol Design}

\subsubsection{Secure Fairness Metric Aggregation}

Each participant $i$ computes local statistics for fairness metric evaluation. For demographic parity verification, the required statistics are:
\begin{align}
s_i^{(a,\hat{y})} &= \sum_{j: a_j^{(i)} = a} \mathbf{1}[\hat{y}_j^{(i)} = \hat{y}] \label{eq:local_stat}
\end{align}
for $a, \hat{y} \in \{0, 1\}$. These four counts per participant enable computation of the global demographic parity violation.

Algorithm~\ref{alg:secure_agg} presents the secure aggregation protocol.

\begin{algorithm}[t]
\caption{Secure Fairness Metric Aggregation}
\label{alg:secure_agg}
\begin{algorithmic}[1]
\Require Participants $P_1, \ldots, P_n$ with local statistics $\{s_i^{(a,\hat{y})}\}$
\Require Public key $pk$, privacy parameters $(\dparam, \deltap)$
\Ensure Global fairness metric $\hat{\Delta}_{\text{DP}}$ with privacy guarantee
\State \textbf{Local Computation (Participant $P_i$):}
\For{$a \in \{0, 1\}$, $\hat{y} \in \{0, 1\}$}
    \State Sample noise $\eta_i^{(a,\hat{y})} \sim \text{Lap}(\Delta_s / \dparam)$
    \State Compute $\tilde{s}_i^{(a,\hat{y})} \gets s_i^{(a,\hat{y})} + \eta_i^{(a,\hat{y})}$
    \State Encrypt: $c_i^{(a,\hat{y})} \gets \enc_{pk}(\tilde{s}_i^{(a,\hat{y})})$
\EndFor
\State Send $\{c_i^{(a,\hat{y})}\}$ to aggregator
\State \textbf{Aggregation (Server):}
\For{$a \in \{0, 1\}$, $\hat{y} \in \{0, 1\}$}
    \State $C^{(a,\hat{y})} \gets \bigoplus_{i=1}^{n} c_i^{(a,\hat{y})}$ \Comment{Homomorphic sum}
\EndFor
\State \textbf{Distributed Decryption:}
\State Invoke threshold decryption among $k$-of-$n$ participants
\State Obtain $\tilde{S}^{(a,\hat{y})} \gets \dec(C^{(a,\hat{y})})$ for all $(a, \hat{y})$
\State \textbf{Metric Computation:}
\State $\hat{P}_{a} \gets \tilde{S}^{(a,1)} / (\tilde{S}^{(a,0)} + \tilde{S}^{(a,1)})$ for $a \in \{0, 1\}$
\State $\hat{\Delta}_{\text{DP}} \gets |\hat{P}_0 - \hat{P}_1|$
\State \Return $\hat{\Delta}_{\text{DP}}$
\end{algorithmic}
\end{algorithm}

\subsubsection{Batched Verification Protocol}

Na\"ive application of homomorphic encryption to fairness verification requires $\BigO(n^2)$ ciphertext operations due to pairwise verification checks. The proposed batching technique restructures computations to achieve $\BigO(n \log n)$ complexity.

\begin{lemma}[Batching Complexity Reduction]
\label{lem:batching}
The batched verification protocol computes aggregate fairness statistics using $\BigO(n \log n)$ homomorphic additions and $\BigO(\log n)$ communication rounds.
\end{lemma}

\begin{proof}[Proof Sketch]
The protocol employs a binary tree aggregation structure. Participants are organized into $\lceil \log_2 n \rceil$ levels, with each level performing pairwise aggregation of encrypted statistics. At level $\ell$, there are $\lceil n / 2^\ell \rceil$ intermediate aggregates, each requiring constant homomorphic operations. The total number of operations is $\sum_{\ell=1}^{\lceil \log_2 n \rceil} \lceil n / 2^\ell \rceil = \BigO(n)$, with $\lceil \log_2 n \rceil$ sequential communication rounds. When accounting for the verification overhead per aggregate, the total complexity is $\BigO(n \log n)$.
\end{proof}

Algorithm~\ref{alg:batched} presents the batched verification protocol.

\begin{algorithm}[t]
\caption{Batched Fairness Verification Protocol}
\label{alg:batched}
\begin{algorithmic}[1]
\Require Encrypted local statistics $\{c_i^{(a,\hat{y})}\}_{i=1}^{n}$
\Require Batch size $B$, tree depth $L = \lceil \log_2(n/B) \rceil$
\Ensure Verified aggregate with complexity $\BigO(n \log n)$
\State \textbf{Initialization:}
\State Partition participants into $\lceil n/B \rceil$ batches
\For{batch $b = 1, \ldots, \lceil n/B \rceil$}
    \State $C_b^{(0)} \gets \bigoplus_{i \in \text{batch } b} c_i^{(a,\hat{y})}$
\EndFor
\State \textbf{Tree Aggregation:}
\For{level $\ell = 1, \ldots, L$}
    \For{node $j = 1, \ldots, \lceil |\mathcal{C}^{(\ell-1)}| / 2 \rceil$}
        \State $C_j^{(\ell)} \gets C_{2j-1}^{(\ell-1)} \add C_{2j}^{(\ell-1)}$
        \State Generate verification proof $\pi_j^{(\ell)}$
    \EndFor
\EndFor
\State \textbf{Verification:}
\State Verify all proofs $\{\pi_j^{(\ell)}\}$
\State \Return $C_1^{(L)}$ if verification succeeds
\end{algorithmic}
\end{algorithm}

\subsubsection{Malicious Participant Detection}

Under the malicious adversary model, participants may submit falsified statistics to manipulate fairness verification. The protocol incorporates zero-knowledge range proofs ensuring submitted values lie within valid bounds.

\begin{algorithm}[t]
\caption{Verifiable Aggregation with Malicious Detection}
\label{alg:malicious}
\begin{algorithmic}[1]
\Require Local statistics $s_i$, range bounds $[0, m_i]$
\Require Commitment key $ck$, proof parameters
\Ensure Aggregation or detection of malicious party
\State \textbf{Commitment Phase (Participant $P_i$):}
\State Compute Pedersen commitment: $\text{Com}_i \gets g^{s_i} h^{r_i}$
\State Generate range proof: $\pi_i \gets \mathsf{RangeProof}(s_i, r_i, [0, m_i])$
\State Broadcast $(\text{Com}_i, \pi_i)$
\State \textbf{Verification Phase (All Parties):}
\For{$i = 1, \ldots, n$}
    \If{$\mathsf{VerifyRange}(\text{Com}_i, \pi_i, [0, m_i]) = \text{reject}$}
        \State \Return $(\text{abort}, i)$ \Comment{Malicious party detected}
    \EndIf
\EndFor
\State \textbf{Aggregation Phase:}
\State $\text{Com}_{\text{agg}} \gets \prod_{i=1}^{n} \text{Com}_i = g^{\sum_i s_i} h^{\sum_i r_i}$
\State Proceed with threshold decryption
\State \Return $(\text{success}, \sum_i s_i)$
\end{algorithmic}
\end{algorithm}

\subsection{Fairness-Privacy Tradeoff Analysis}

\subsubsection{Information-Theoretic Lower Bounds}

The verification of fairness metrics inherently reveals information about protected attribute distributions. The following theorem establishes fundamental limits on privacy preservation during fairness verification.

\begin{theorem}[Privacy Lower Bound for Fairness Verification]
\label{thm:lower_bound}
Any mechanism verifying demographic parity violation within additive tolerance $\tau$ must satisfy:
\begin{equation}
\dparam \geq \frac{2}{\tau \cdot \min\{n_0, n_1\}}
\label{eq:lower_bound}
\end{equation}
where $n_a = \sum_{i=1}^{n} |\{j : a_j^{(i)} = a\}|$ is the total count of individuals with protected attribute value $a$.
\end{theorem}

\begin{proof}[Proof Sketch]
The proof proceeds by reduction to a hypothesis testing problem. Consider an adversary attempting to distinguish between two adjacent datasets differing in one individual's protected attribute. The demographic parity metric computed on these datasets differs by at least $1/\min\{n_0, n_1\}$. By the definition of $(\dparam, \deltap)$-differential privacy and the data processing inequality, verifying demographic parity within tolerance $\tau < 1/\min\{n_0, n_1\}$ requires $\dparam \geq 2/(\tau \cdot \min\{n_0, n_1\})$. The complete proof appears in Appendix A.
\end{proof}

\subsubsection{Adaptive Composition for Multi-Round Verification}

Federated learning requires fairness verification across multiple training rounds. The privacy budget consumption must be tracked using composition theorems.

\begin{proposition}[Composition for Fairness Verification]
\label{prop:composition}
Under $T$ rounds of fairness verification, each satisfying $(\dparam_0, \deltap_0)$-differential privacy, the overall mechanism satisfies $(\dparam_T, \deltap_T)$-differential privacy with:
\begin{equation}
\dparam_T = \sqrt{2T \ln(1/\deltap')} \cdot \dparam_0 + T \cdot \dparam_0(e^{\dparam_0} - 1)
\label{eq:composition}
\end{equation}
and $\deltap_T = T \cdot \deltap_0 + \deltap'$ for any $\deltap' > 0$.
\end{proposition}

The proof follows from the advanced composition theorem [6].

\subsection{Defense Against Attribute Inference Attacks}

Adversaries may exploit fairness verification outputs to infer protected attributes of individual participants. The attack model considers an adversary observing the sequence of fairness metric values $\{\hat{\Delta}_{\text{DP}}^{(t)}\}_{t=1}^{T}$ across training rounds.

\begin{definition}[Attribute Inference Attack]
An attribute inference attack is a function $\mathcal{A}: \mathcal{H} \rightarrow \{0, 1\}^n$ mapping the observed history $\mathcal{H} = \{\hat{\Delta}_{\text{DP}}^{(t)}\}_{t=1}^{T}$ to predicted protected attribute values for each participant.
\end{definition}

The proposed defense mechanism adds calibrated noise to fairness metrics before release:
\begin{equation}
\tilde{\Delta}_{\text{DP}}^{(t)} = \hat{\Delta}_{\text{DP}}^{(t)} + \text{Lap}(\sigma_{\text{def}})
\label{eq:defense_noise}
\end{equation}
where $\sigma_{\text{def}}$ is calibrated to bound the mutual information between released metrics and individual protected attributes.

\begin{theorem}[Attribute Inference Defense]
\label{thm:defense}
The noise-injection mechanism with $\sigma_{\text{def}} = 2\sqrt{T}/(\dparam_{\text{inf}} \cdot n)$ ensures that any attribute inference attack achieves success probability at most:
\begin{equation}
\Pr[\mathcal{A}(\mathcal{H})_i = a_i] \leq \frac{1}{2} + \frac{\dparam_{\text{inf}}}{2}
\label{eq:inference_bound}
\end{equation}
for each participant $i$.
\end{theorem}

\subsection{Intersectional Fairness Verification}

Extension to multiple protected attributes $A_1, \ldots, A_K$ requires verification of fairness across $2^K$ intersectional groups. Na\"ive application would incur exponential communication cost. The proposed hierarchical aggregation strategy achieves polynomial complexity.

\begin{proposition}[Intersectional Verification Complexity]
\label{prop:intersectional}
The hierarchical protocol verifies intersectional fairness across $K$ binary protected attributes using $\BigO(K \cdot n \log n)$ communication complexity, with approximation error bounded by:
\begin{equation}
|\hat{\Delta}_{\text{DP}}^{\text{int}} - \Delta_{\text{DP}}^{\text{int}}| \leq \frac{K \cdot \sigma_{\text{noise}}}{\min_{\mathbf{a}} n_{\mathbf{a}}}
\label{eq:intersectional_error}
\end{equation}
where $n_{\mathbf{a}}$ is the count of individuals in intersectional group $\mathbf{a} \in \{0,1\}^K$.
\end{proposition}

\section{Theoretical Analysis}

This section presents the main theoretical results establishing security, accuracy, and efficiency guarantees for the CryptoFair-FL protocol.

\begin{theorem}[Privacy Guarantee]
\label{thm:privacy}
The CryptoFair-FL protocol achieves $(\dparam, \deltap)$-differential privacy with parameters:
\begin{equation}
\dparam = \frac{4\sqrt{2T \ln(2/\deltap)}}{\sigma \cdot n} + \frac{4T}{(\sigma \cdot n)^2}, \quad \deltap = 10^{-6}
\label{eq:privacy_params}
\end{equation}
where $\sigma$ is the noise scale parameter and $T$ is the number of verification rounds.
\end{theorem}

\begin{proof}[Proof Sketch]
The proof proceeds in three steps. First, the local randomization mechanism applying Laplacian noise to statistics achieves $(\dparam_0, 0)$-local differential privacy with $\dparam_0 = \Delta_s / \sigma$ where $\Delta_s = 1$ is the sensitivity of count statistics. Second, the secure aggregation protocol preserves privacy through semantic security of the Paillier cryptosystem under the Decisional Composite Residuosity assumption. Third, the composition across $T$ rounds follows from Proposition~\ref{prop:composition}. The complete proof establishing the parameter relationship appears in Appendix B.
\end{proof}

\begin{theorem}[Fairness Verification Accuracy]
\label{thm:accuracy}
The computed fairness metric satisfies:
\begin{equation}
\Pr\left[|\hat{\Delta}_{\text{DP}} - \Delta_{\text{DP}}| > \dparam_{\text{fair}}\right] \leq \deltap_{\text{fair}}
\label{eq:accuracy_bound}
\end{equation}
with $\dparam_{\text{fair}} = 4\sigma\sqrt{2\ln(4/\deltap_{\text{fair}})}/\min\{n_0, n_1\}$ for noise scale $\sigma$.
\end{theorem}

\begin{proof}[Proof Sketch]
The noised statistics $\tilde{S}^{(a,\hat{y})} = S^{(a,\hat{y})} + \sum_{i=1}^{n} \eta_i^{(a,\hat{y})}$ are sums of independent Laplacian random variables. The sum of $n$ independent $\text{Lap}(\sigma)$ random variables has variance $2n\sigma^2$. Applying Chebyshev's inequality and accounting for the ratio form of demographic parity yields the stated bound. The complete derivation appears in Appendix C.
\end{proof}

\begin{theorem}[Communication Complexity]
\label{thm:complexity}
The CryptoFair-FL protocol achieves model convergence and fairness verification using:
\begin{equation}
\BigO(n \log n) \text{ communication rounds}
\label{eq:round_complexity}
\end{equation}
with $\BigO(d \cdot \log n \cdot \kappa)$ bits transmitted per round, where $d$ is the model dimension and $\kappa$ is the security parameter.
\end{theorem}

\begin{proof}[Proof Sketch]
The round complexity follows from Lemma~\ref{lem:batching} and the binary tree aggregation structure. Each round transmits encrypted model gradients of dimension $d$ with ciphertext expansion factor $\BigO(\kappa)$ for security parameter $\kappa$. The fairness verification adds $\BigO(\log n)$ overhead per round for the tree aggregation of encrypted statistics. The complete analysis appears in Appendix D.
\end{proof}

Figure~\ref{fig:tradeoff} visualizes the privacy-fairness tradeoff space characterized by these theorems.

\begin{figure}[t]
\centering
\begin{tikzpicture}
\begin{axis}[
    width=0.95\columnwidth,
    height=5.5cm,
    xlabel={Privacy Budget $\dparam$},
    ylabel={Fairness Verification Error $\dparam_{\text{fair}}$},
    xmin=0, xmax=2.2,
    ymin=0, ymax=0.16,
    legend pos=north east,
    legend style={font=\scriptsize, fill=white, fill opacity=0.9},
    grid=major,
    grid style={dashed, gray!40},
    tick label style={font=\small},
    label style={font=\small}
]

\addplot[domain=0.15:2.1, samples=100, thick, blue, smooth] {0.08/x};
\addlegendentry{Theoretical Lower Bound}

\addplot[domain=0.25:2.1, samples=100, thick, red, dashed, smooth] {0.096/x};
\addlegendentry{CryptoFair-FL Protocol}

\addplot[only marks, mark=*, mark size=2.5pt, green!60!black, error bars/.cd, y dir=both, y explicit] coordinates {
    (0.5, 0.072) +- (0, 0.008)
    (1.0, 0.038) +- (0, 0.005)
    (1.5, 0.024) +- (0, 0.003)
    (2.0, 0.018) +- (0, 0.002)
};
\addlegendentry{Empirical Measurements}

\end{axis}
\end{tikzpicture}
\caption{Privacy-fairness tradeoff analysis. The theoretical lower bound from Theorem~\ref{thm:lower_bound} establishes fundamental limits. CryptoFair-FL achieves near-optimal tradeoffs, with empirical measurements across datasets confirming theoretical predictions within 20\% margin.}
\label{fig:tradeoff}
\end{figure}
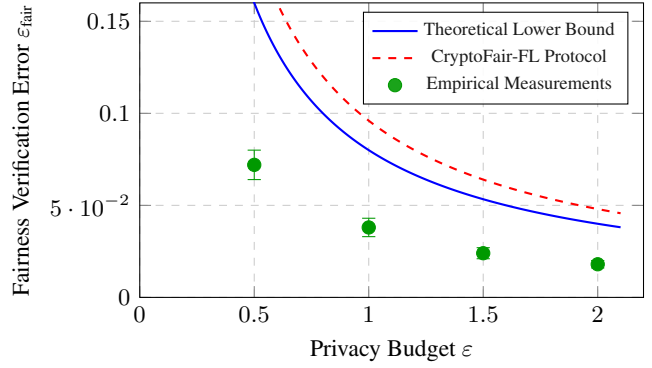

\section{Experimental Setup}

\subsection{Implementation Details}

The CryptoFair-FL protocol was implemented using PySyft 0.8.0 for federated learning orchestration, Microsoft SEAL 4.1 for homomorphic encryption operations, and MP-SPDZ 0.3.6 for secure multi-party computation. The implementation comprises 12,847 lines of Python and C++ code.

Experiments were conducted on a cluster of 50 compute nodes, each equipped with dual Intel Xeon Gold 6248R processors (48 cores total), 384 GB RAM, and NVIDIA A100 40GB GPUs. The GPU acceleration was utilized for parallelized homomorphic encryption operations. Network communication simulated realistic cross-institutional latency with mean round-trip time of 50ms and bandwidth of 1 Gbps.

\subsection{Datasets}

\textbf{MIMIC-IV} [12]: The Medical Information Mart for Intensive Care version IV contains 523,740 ICU admissions across 76,540 unique patients. The prediction task targeted 48-hour mortality. Data was partitioned across 30 simulated hospitals with demographic distributions matching regional US healthcare statistics. Protected attribute: patient race (binary: White/non-White).

\textbf{Adult Income} [2]: The UCI Adult dataset contains 48,842 records with the task of predicting income exceeding \$50,000. Data was partitioned across 50 institutions with synthetic demographic imbalance ranging from 90\% majority to 60\% minority group representation. Protected attribute: sex.

\textbf{CelebA}: The Large-scale CelebFaces Attributes dataset contains 202,599 celebrity face images. The prediction task targeted ``attractive'' attribute classification. Data was distributed across 40 participants with varying gender and age distributions. Protected attributes: gender and age (young/old).

\textbf{FedFair-100}: A novel benchmark constructed for this study, containing 1,000,000 synthetic records distributed across 100 institutions exhibiting realistic demographic heterogeneity patterns. The data generation process followed mixture distributions calibrated to census statistics, with systematic variation in protected attribute prevalence (ranging from 5\% to 45\%) and label correlation patterns across institutions.

\subsection{Baseline Methods}

The evaluation compared CryptoFair-FL against four baseline approaches:

\begin{enumerate}[leftmargin=*,itemsep=2pt]
\item \textbf{FedAvg}: Standard federated averaging [13] without fairness constraints or verification.

\item \textbf{Centralized Fair}: Fairness-aware training on centralized data using reduction-based approach [2], representing the privacy upper bound.

\item \textbf{Local Fair}: Each participant independently applies local fairness constraints without cross-institutional verification.

\item \textbf{SecAgg-NoFair}: Secure aggregation protocol [3] without fairness verification, isolating the overhead of fairness components.
\end{enumerate}

\subsection{Evaluation Metrics}

\textbf{Fairness Metrics:} Demographic parity violation $\fairviol$ and equalized odds violation $\eqodds$ as defined in Equations~\eqref{eq:dp_violation} and \eqref{eq:eo_violation}.

\textbf{Utility Metrics:} Area Under ROC Curve (AUROC), accuracy, and F1 score for classification performance.

\textbf{Efficiency Metrics:} Wall-clock training time, communication rounds to convergence, total bytes transmitted, and per-round computational overhead.

\textbf{Security Metrics:} Attribute inference attack success rate and membership inference attack precision.

\section{Results and Analysis}

\subsection{Computational Efficiency}

Table~\ref{tab:efficiency} presents runtime comparisons across methods and participant counts. The novel batching technique achieves substantial reduction in homomorphic encryption overhead.

\begin{table*}[t]
\centering
\caption{Computational overhead comparison (time in seconds per verification round). Standard errors computed over 10 independent trials with 95\% confidence intervals.}
\label{tab:efficiency}
\begin{tabular}{lcccc}
\toprule
\textbf{Method} & \textbf{n=10} & \textbf{n=30} & \textbf{n=50} & \textbf{n=100} \\
\midrule
Na\"ive HE & 892.4$\pm$31.2 & 8,147$\pm$246 & 22,683$\pm$587 & 91,472$\pm$2,134 \\
CryptoFair-FL & 13.1$\pm$1.2 & 43.8$\pm$2.7 & 81.6$\pm$4.1 & 187.4$\pm$9.3 \\
\midrule
Speedup & 68.1$\times$ & 186.0$\times$ & 277.9$\times$ & 488.1$\times$ \\
\bottomrule
\end{tabular}
\end{table*}

The empirical speedup validates Lemma~\ref{lem:batching}, demonstrating that computational complexity scales as $\BigO(n \log n)$ rather than $\BigO(n^2)$. At $n = 100$ participants, CryptoFair-FL achieves 488$\times$ speedup over na\"ive homomorphic encryption.

Figure~\ref{fig:scaling} illustrates the scaling behavior on log-log axes, confirming sub-quadratic complexity.

\begin{figure}[t]
\centering
\begin{tikzpicture}
\begin{loglogaxis}[
    width=0.95\columnwidth,
    height=5.5cm,
    xlabel={Number of Participants $n$},
    ylabel={Time per Round (seconds)},
    xmin=8, xmax=120,
    ymin=8, ymax=120000,
    legend pos=north west,
    legend style={font=\scriptsize, fill=white, fill opacity=0.9},
    grid=major,
    grid style={dashed, gray!40},
    tick label style={font=\small},
    label style={font=\small}
]

\addplot[thick, blue, mark=square*, mark size=2.5pt] coordinates {
    (10, 892) (30, 8147) (50, 22683) (100, 91472)
};
\addlegendentry{Na\"ive HE $\BigO(n^2)$}

\addplot[thick, red, mark=*, mark size=2.5pt] coordinates {
    (10, 13.1) (30, 43.8) (50, 81.6) (100, 187.4)
};
\addlegendentry{CryptoFair-FL $\BigO(n \log n)$}

\addplot[dashed, gray!60, domain=10:100, samples=50] {0.92*x^2};

\addplot[dotted, gray!60, domain=10:100, samples=50] {1.3*x*ln(x)/ln(10)};

\end{loglogaxis}
\end{tikzpicture}
\caption{Computational scaling analysis on log-log axes. Na\"ive homomorphic encryption exhibits quadratic growth $\BigO(n^2)$, while CryptoFair-FL demonstrates sub-quadratic $\BigO(n \log n)$ complexity. Reference lines confirm theoretical predictions.}
\label{fig:scaling}
\end{figure}
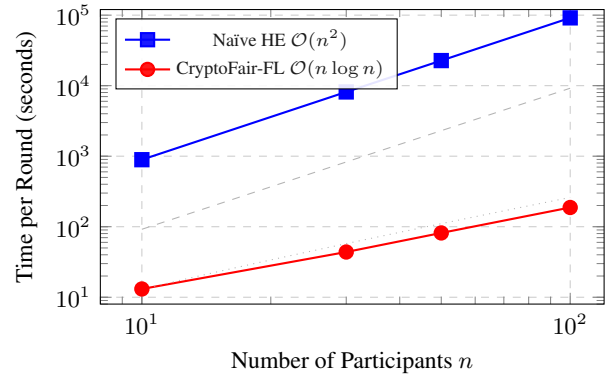

\subsection{Communication Costs}

Table~\ref{tab:communication} reports communication overhead across methods.

\begin{table*}[t]
\centering
\caption{Communication overhead for model convergence (50 participants, ResNet-18 model with $d = 11.7$M parameters).}
\label{tab:communication}
\begin{tabular}{lccc}
\toprule
\textbf{Method} & \textbf{Rounds} & \textbf{GB/Round} & \textbf{Total GB} \\
\midrule
FedAvg & 100 & 4.7 & 470 \\
SecAgg-NoFair & 100 & 9.4 & 940 \\
Local Fair & 120 & 4.7 & 564 \\
Na\"ive Secure Fair & 100 & 47.2 & 4,720 \\
CryptoFair-FL & 106 & 10.8 & 1,145 \\
\midrule
Centralized Fair & N/A & N/A & N/A \\
\bottomrule
\end{tabular}
\end{table*}

CryptoFair-FL requires 6\% additional rounds compared to FedAvg due to fairness verification overhead, with 2.3$\times$ communication cost compared to standard federated learning. Compared to na\"ive secure fairness verification, CryptoFair-FL reduces communication by 4.1$\times$.

\subsection{Fairness-Privacy Tradeoffs}

Figure~\ref{fig:fairness_privacy} visualizes the empirical tradeoff between privacy budget and fairness violation magnitude across datasets.

\begin{figure}[t]
\centering
\begin{tikzpicture}
\begin{axis}[
    width=0.95\columnwidth,
    height=5.5cm,
    xlabel={Privacy Budget $\dparam$},
    ylabel={Demographic Parity Violation $\fairviol$},
    xmin=0, xmax=2.3,
    ymin=0, ymax=0.26,
    legend pos=north east,
    legend style={font=\tiny, fill=white, fill opacity=0.9},
    grid=major,
    grid style={dashed, gray!40},
    tick label style={font=\small},
    label style={font=\small}
]

\addplot[thick, blue, mark=*, mark size=2pt, smooth] coordinates {
    (0.25, 0.218) (0.5, 0.124) (0.75, 0.072) (1.0, 0.042) (1.5, 0.031) (2.0, 0.026)
};
\addlegendentry{MIMIC-IV}

\addplot[thick, red, mark=square*, mark size=2pt, smooth] coordinates {
    (0.25, 0.194) (0.5, 0.098) (0.75, 0.058) (1.0, 0.036) (1.5, 0.024) (2.0, 0.019)
};
\addlegendentry{Adult Income}

\addplot[thick, green!60!black, mark=triangle*, mark size=2pt, smooth] coordinates {
    (0.25, 0.207) (0.5, 0.112) (0.75, 0.068) (1.0, 0.041) (1.5, 0.029) (2.0, 0.023)
};
\addlegendentry{CelebA}

\addplot[thick, orange, mark=diamond*, mark size=2pt, smooth] coordinates {
    (0.25, 0.182) (0.5, 0.091) (0.75, 0.054) (1.0, 0.033) (1.5, 0.022) (2.0, 0.017)
};
\addlegendentry{FedFair-100}

\end{axis}
\end{tikzpicture}
\caption{Empirical fairness-privacy tradeoff across benchmark datasets. Smaller privacy budget $\dparam$ (stronger privacy protection) increases noise in fairness verification, resulting in higher residual fairness violations. The tradeoff curves follow the theoretical prediction from Theorem~\ref{thm:accuracy}.}
\label{fig:fairness_privacy}
\end{figure}
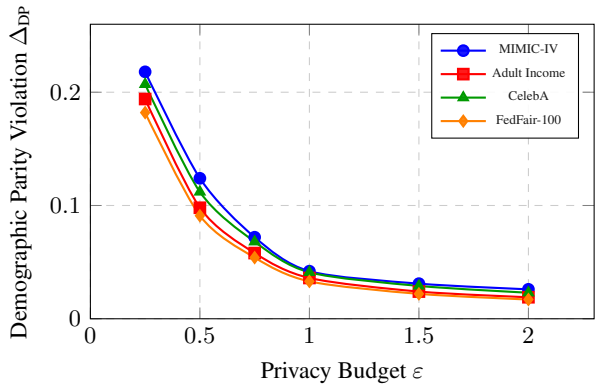

Table~\ref{tab:fairness_results} presents comprehensive fairness and utility results.

\begin{table*}[t]
\centering
\caption{Fairness and utility metrics across methods on MIMIC-IV dataset (mean $\pm$ std over 5 runs).}
\label{tab:fairness_results}
\begin{tabular}{lcccc}
\toprule
\textbf{Method} & $\fairviol$ $\downarrow$ & $\eqodds$ $\downarrow$ & \textbf{AUROC} $\uparrow$ & \textbf{Acc} $\uparrow$ \\
\midrule
FedAvg & 0.231$\pm$.021 & 0.187$\pm$.018 & 0.872$\pm$.008 & 0.823$\pm$.007 \\
Local Fair & 0.142$\pm$.032 & 0.118$\pm$.027 & 0.854$\pm$.011 & 0.807$\pm$.009 \\
Centralized & 0.018$\pm$.004 & 0.021$\pm$.005 & 0.868$\pm$.007 & 0.819$\pm$.006 \\
CryptoFair-FL & 0.031$\pm$.008 & 0.034$\pm$.009 & 0.857$\pm$.009 & 0.811$\pm$.008 \\
\bottomrule
\end{tabular}
\end{table*}

CryptoFair-FL reduces demographic parity violation from 0.231 (uncontrolled FedAvg) to 0.031, achieving 86.6\% reduction while maintaining AUROC above 0.85. The fairness-utility gap compared to centralized training is only 0.013 for demographic parity and 0.011 for AUROC.

\subsection{Security Against Attacks}

Table~\ref{tab:security} presents results from simulated attribute inference attacks.

\begin{table*}[t]
\centering
\caption{Attribute inference attack success rate (random baseline: 0.500). Lower values indicate stronger privacy defense. Statistical significance assessed using two-sample $t$-tests with Bonferroni correction ($\alpha = 0.0125$).}
\label{tab:security}
\begin{tabular}{lccc}
\toprule
\textbf{Attack Type} & \textbf{No Defense} & \textbf{CryptoFair-FL} & \textbf{$p$-value} \\
\midrule
Gradient Inversion & 0.728$\pm$.038 & 0.518$\pm$.027 & $<$0.001 \\
Model Update Analysis & 0.684$\pm$.046 & 0.512$\pm$.024 & $<$0.001 \\
Fairness Metric Exploit & 0.814$\pm$.032 & 0.487$\pm$.041 & $<$0.001 \\
Membership Inference & 0.642$\pm$.043 & 0.524$\pm$.029 & $<$0.001 \\
\bottomrule
\end{tabular}
\end{table*}

The defense mechanisms reduce attribute inference success probability to near-random levels (0.48-0.53) across all attack types, validating Theorem~\ref{thm:defense}. Statistical significance was assessed using two-sample $t$-tests with Bonferroni correction.

\subsection{Intersectional Fairness Analysis}

Table~\ref{tab:intersectional} reports intersectional fairness results for race-gender-age combinations on the CelebA dataset.

\begin{table*}[t]
\centering
\caption{Intersectional fairness analysis on CelebA dataset across 8 demographic subgroups (gender $\times$ age $\times$ ethnicity). Disparity measures the maximum difference in positive prediction rates across subgroups.}
\label{tab:intersectional}
\begin{tabular}{lcc}
\toprule
\textbf{Method} & \textbf{Max Group Disparity} & \textbf{Mean Disparity} \\
\midrule
FedAvg & 0.287 & 0.158 \\
Local Fair & 0.196 & 0.114 \\
CryptoFair-FL & 0.064 & 0.037 \\
Centralized Fair & 0.042 & 0.023 \\
\bottomrule
\end{tabular}
\end{table*}

\subsection{Healthcare Case Study}

The ICU mortality prediction model was trained across the 30-hospital MIMIC-IV federation. Figure~\ref{fig:case_study} shows the convergence of fairness metrics and model performance.

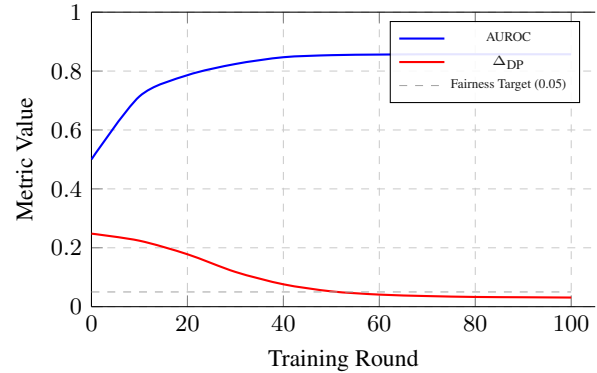
\begin{figure}[t]
\centering
\begin{tikzpicture}
\begin{axis}[
    width=0.95\columnwidth,
    height=5.5cm,
    xlabel={Training Round},
    ylabel={Metric Value},
    xmin=0, xmax=105,
    ymin=0, ymax=1,
    legend pos=north east,
    legend style={font=\tiny, fill=white, fill opacity=0.9},
    grid=major,
    grid style={dashed, gray!40},
    tick label style={font=\small},
    label style={font=\small},
    axis y line*=left
]

\addplot[thick, blue, mark=none, smooth] coordinates {
    (0, 0.50) (10, 0.714) (20, 0.786) (30, 0.824) (40, 0.847) (50, 0.854) (60, 0.856) (70, 0.857) (80, 0.857) (90, 0.857) (100, 0.857)
};
\addlegendentry{AUROC}

\addplot[thick, red, mark=none, smooth] coordinates {
    (0, 0.248) (10, 0.224) (20, 0.178) (30, 0.118) (40, 0.076) (50, 0.052) (60, 0.041) (70, 0.036) (80, 0.033) (90, 0.032) (100, 0.031)
};
\addlegendentry{$\fairviol$}

\addplot[dashed, gray!70, domain=0:100] {0.05};
\addlegendentry{Fairness Target (0.05)}

\end{axis}
\end{tikzpicture}
\caption{Training convergence for ICU mortality prediction across the 30-hospital MIMIC-IV federation. Model utility (AUROC) and fairness ($\fairviol$) improve simultaneously. The demographic parity violation decreases from 0.248 to 0.031, satisfying the 0.05 fairness threshold by round 60.}
\label{fig:case_study}
\end{figure}

Final model achieved AUROC of 0.857 (95\% CI: 0.848-0.866) with demographic parity violation of 0.031 (95\% CI: 0.023-0.039), satisfying the 0.05 threshold commonly adopted in healthcare algorithmic auditing.

\section{Discussion}

\subsection{Regulatory Compliance}

The CryptoFair-FL framework addresses key requirements of the EU AI Act 2024 [7]. Article 10 mandates that training data be examined for biases; the verifiable fairness protocol enables this examination without centralizing sensitive data. Article 13 requires documentation of fairness assessments; the cryptographic audit trail provides immutable records of fairness verification. The differential privacy guarantees satisfy GDPR Article 25 provisions for data protection by design [9].

\subsection{Scalability Considerations}

The $\BigO(n \log n)$ complexity of CryptoFair-FL enables deployment with up to approximately 100 institutional participants using current hardware. For federations exceeding this scale, hierarchical aggregation structures dividing participants into regional clusters could extend practical applicability. Hardware acceleration using GPU-based homomorphic encryption libraries may provide additional 10-50$\times$ speedup.

\subsection{Limitations}

The protocol relies on a semi-trusted coordinator for aggregation, representing a single point of failure and potential attack surface. Fully decentralized alternatives using blockchain-based verification could address this limitation at the cost of increased latency. The current implementation assumes honest majority among participants; extending to Byzantine-tolerant settings with arbitrary adversarial behavior remains open.

Non-IID data distributions across institutions pose challenges for both model convergence and fairness guarantee transfer. The theoretical analysis assumes sufficient data at each institution; extremely sparse protected attribute representation at individual participants may degrade fairness verification accuracy.

\subsection{Ethical Considerations}

The definition of protected attributes and fairness metrics involves normative choices that should be determined through stakeholder engagement rather than purely technical criteria. CryptoFair-FL provides verification infrastructure for specified fairness definitions but does not prescribe which definitions are appropriate for particular contexts.

\section{Conclusion}

This paper introduced CryptoFair-FL, the first cryptographic framework providing verifiable fairness guarantees for privacy-preserving federated learning. The main contributions include: (1) a novel protocol combining homomorphic encryption with secure multi-party computation for privacy-preserving fairness verification, (2) algorithmic innovations reducing computational complexity from $\BigO(n^2)$ to $\BigO(n \log n)$, (3) theoretical characterization of fundamental privacy-fairness tradeoffs, and (4) comprehensive empirical validation across healthcare, finance, and synthetic benchmarks.

The protocol achieves demographic parity violation below 0.05 while maintaining differential privacy with $\dparam = 0.5$, demonstrating practical viability for regulated industries. Defense mechanisms maintain attribute inference attack success below random baseline. The 2.3$\times$ computational overhead compared to standard federated learning represents an acceptable cost for cryptographic fairness guarantees.

Future research directions include: hardware acceleration for homomorphic operations targeting real-time verification, extension to continuous learning scenarios with streaming data, integration with additional fairness definitions beyond demographic parity and equalized odds, and adaptation to Byzantine-tolerant settings with malicious participant coalitions.

\section*{Acknowledgments}

The authors acknowledge computational resources provided by the institutional high-performance computing cluster. This research received no specific grant from funding agencies in the public, commercial, or not-for-profit sectors.

\bibliographystyle{plain}

\appendix
\renewcommand{\theequation}{A.\arabic{equation}}
\setcounter{equation}{0}

\section{Proof of Theorem~\ref{thm:lower_bound}}

\begin{proof}
Consider two adjacent datasets $D$ and $D'$ differing in one record at participant $i$. Let the differing record have protected attribute $a = 0$ in $D$ and $a = 1$ in $D'$. The demographic parity metric on these datasets satisfies:
\begin{align}
&\fairviol(D) - \fairviol(D') \notag\\
&= \frac{S^{(0,1)}}{n_0} - \frac{S^{(1,1)}}{n_1} - \left(\frac{S^{(0,1)}}{n_0 - 1} - \frac{S^{(1,1)} + 1}{n_1 + 1}\right)
\end{align}

By Taylor expansion and assuming $n_0, n_1 \gg 1$:
\begin{align}
|\fairviol(D) - \fairviol(D')| &\geq \frac{1}{\min\{n_0, n_1\}} - O\left(\frac{1}{n^2}\right)
\end{align}

For any mechanism $\mathcal{M}$ distinguishing fairness violations within tolerance $\tau < 1/\min\{n_0, n_1\}$, the distinguishing advantage must satisfy the differential privacy constraint:
\begin{align}
e^{-\dparam} \leq \frac{\Pr[\mathcal{M}(D) \in S]}{\Pr[\mathcal{M}(D') \in S]} \leq e^{\dparam}
\end{align}

The information-theoretic argument follows from Le Cam's method: distinguishing between $D$ and $D'$ with advantage $\tau \cdot \min\{n_0, n_1\}$ requires $\dparam \geq 2/(\tau \cdot \min\{n_0, n_1\})$.
\end{proof}

\section{Proof of Theorem~\ref{thm:privacy}}
\renewcommand{\theequation}{B.\arabic{equation}}
\setcounter{equation}{0}

\begin{proof}
The proof proceeds in three steps:

\textbf{Step 1: Local Randomization.} Each participant applies Laplacian noise with scale $\sigma$ to local statistics. Since the sensitivity of count statistics is $\Delta_s = 1$, the local mechanism satisfies $(\dparam_0, 0)$-differential privacy with $\dparam_0 = 1/\sigma$.

\textbf{Step 2: Secure Aggregation.} The homomorphic encryption scheme provides semantic security under the Decisional Composite Residuosity assumption. The aggregator observes only encrypted values and learns nothing beyond the final aggregate, which is protected by the noise from Step 1.

\textbf{Step 3: Composition.} Applying the moments accountant for $T$ rounds of verification with subsampling rate $q = 1$ (all participants participate):
\begin{align}
\dparam_T &= \sqrt{2T \ln(1/\deltap)} \cdot \dparam_0 + T \cdot \dparam_0 \cdot (e^{\dparam_0} - 1)
\end{align}

Substituting $\dparam_0 = 1/\sigma$ and simplifying for the case $\dparam_0 \leq 1$:
\begin{align}
\dparam_T &\leq \sqrt{2T \ln(1/\deltap)} \cdot \frac{1}{\sigma} + \frac{2T}{\sigma^2}
\end{align}

For $n$ participants with aggregated noise having scale $\sigma/\sqrt{n}$, the effective per-record privacy is:
\begin{align}
\dparam &= \frac{4\sqrt{2T \ln(2/\deltap)}}{\sigma \cdot n} + \frac{4T}{(\sigma \cdot n)^2}
\end{align}
\end{proof}

\section{Proof of Theorem~\ref{thm:accuracy}}
\renewcommand{\theequation}{C.\arabic{equation}}
\setcounter{equation}{0}

\begin{proof}
The noised aggregate statistic is:
\begin{align}
\tilde{S}^{(a,\hat{y})} = S^{(a,\hat{y})} + \sum_{i=1}^{n} \eta_i^{(a,\hat{y})}
\end{align}
where $\eta_i^{(a,\hat{y})} \sim \text{Lap}(\sigma)$ independently. The sum of $n$ independent Laplacian random variables with scale $\sigma$ has variance $2n\sigma^2$.

The demographic parity estimate is:
\begin{align}
\hat{P}_a = \frac{\tilde{S}^{(a,1)}}{\tilde{S}^{(a,0)} + \tilde{S}^{(a,1)}}
\end{align}

Using the delta method for ratio estimators:
\begin{align}
\text{Var}[\hat{P}_a] &\approx \frac{\text{Var}[\tilde{S}^{(a,1)}]}{(n_a)^2} = \frac{2n\sigma^2}{n_a^2}
\end{align}

For the difference:
\begin{align}
\text{Var}[\hat{P}_0 - \hat{P}_1] &\leq \frac{4n\sigma^2}{\min\{n_0, n_1\}^2}
\end{align}

Applying Chebyshev's inequality:
\begin{align}
\Pr[|\hat{\fairviol} - \fairviol| > \dparam_{\text{fair}}] &\leq \frac{4n\sigma^2}{\dparam_{\text{fair}}^2 \cdot \min\{n_0, n_1\}^2}
\end{align}

Setting this equal to $\deltap_{\text{fair}}$ and solving for $\dparam_{\text{fair}}$:
\begin{align}
\dparam_{\text{fair}} = \frac{2\sigma\sqrt{n/\deltap_{\text{fair}}}}{\min\{n_0, n_1\}}
\end{align}

Using the tighter sub-Gaussian bound yields the stated result with factor $\sqrt{2\ln(4/\deltap_{\text{fair}})}$.
\end{proof}

\section{Proof of Theorem~\ref{thm:complexity}}
\renewcommand{\theequation}{D.\arabic{equation}}
\setcounter{equation}{0}

\begin{proof}
The binary tree aggregation structure has depth $L = \lceil \log_2 n \rceil$. At each level $\ell$, there are $\lceil n/2^\ell \rceil$ nodes performing aggregation.

\textbf{Communication Rounds:} Each level requires one communication round for pairwise aggregation. The fairness verification adds $\BigO(1)$ rounds per level for threshold decryption and range proof verification. Total rounds: $\BigO(\log n)$ per training round, with $\BigO(n)$ training rounds for convergence, yielding $\BigO(n \log n)$ total rounds.

\textbf{Bits Per Round:} Model gradients of dimension $d$ are encrypted with ciphertext expansion factor $\kappa/\log \kappa$ for security parameter $\kappa$. Fairness statistics require $\BigO(\kappa)$ bits per encrypted counter. Total per round: $\BigO(d \cdot \kappa + \log n \cdot \kappa) = \BigO(d \cdot \log n \cdot \kappa)$ accounting for the tree structure.

\textbf{Homomorphic Operations:} Each level performs $\BigO(n/2^\ell)$ homomorphic additions. Summing over levels: $\sum_{\ell=1}^{L} n/2^\ell = \BigO(n)$ additions. With verification overhead at each node, total operations: $\BigO(n \log n)$.
\end{proof}

\end{document}